\documentclass[11pt]{article}

\usepackage[margin=1in]{geometry}
\usepackage[utf8]{inputenc}
\usepackage{soul}
\usepackage{url}
\usepackage{amsmath,amsfonts,amsthm}
\usepackage{thmtools}
\usepackage{thm-restate} 

\allowdisplaybreaks

\usepackage{booktabs}
\usepackage{color}
\usepackage{fancybox}
\usepackage{amssymb}
\usepackage{float}
\usepackage{dsfont}
\usepackage{bm}
\usepackage{bbm}
\usepackage{hyperref,cleveref}
\usepackage{subfig}

\usepackage{xspace} 
\usepackage{enumerate}
\usepackage{graphicx}

\usepackage{comment}

\newtheorem{theorem}{Theorem}

\newtheorem{lemma}{Lemma}

 \usepackage{array}

\newtheorem{proposition}{Proposition}
\usepackage{algorithm}
\usepackage[noend]{algorithmic}

\global\long\def\N{\mathbb{N}}%
\global\long\def\OPT{\mathrm{OPT}}%

\newcommand{\eps}{\varepsilon}

\newcommand{\bI}{\bar{I}}
\newcommand{\hI}{\hat{I}}
\newcommand{\calL}{{\mathcal L}}

\title{A simpler QPTAS for scheduling jobs with precedence constraints}
\author{
Syamantak Das\\ IIIT-Delhi, India \\
\texttt{syamantak@iiitd.ac.in}
        \and 
        Andreas Wiese\\ Technical University of Munich\\
        \texttt{andreas.wiese@tum.de}        
}
\begin{document}
    \maketitle

\begin{abstract}
		We study the classical scheduling problem of minimizing the makespan
		of a set of unit size jobs with
		precedence constraints on parallel identical machines. Research on the problem dates back to the
		landmark paper by Graham from 1966 who showed that the simple List
		Scheduling algorithm is a $(2-\frac{1}{m})$-approximation. Interestingly,
		it is open whether the problem is NP-hard if $m=3$ which is one of
		the few remaining open problems in the seminal book by Garey and Johnson.
		Recently, quite some progress has been made for the setting that $m$
		is a constant. In a break-through paper, Levey and Rothvoss presented
		a $(1+\epsilon)$-approximation with a running time of $n^{(\log n)^{O((m^{2}/\epsilon^{2})\log\log n)}}$~{[}STOC
		2016, SICOMP 2019{]} and this running time was improved to quasi-polynomial
		by Garg {[}ICALP 2018{]} and to even $n^{O_{m,\epsilon}(\log^{3}\log n)}$
		by Li {[}SODA 2021{]}. These results use techniques like LP-hierarchies,
		conditioning on certain well-selected jobs, and abstractions like
		(partial) dyadic systems and virtually valid schedules. 
		
		In this paper, we present a QPTAS for the problem which is arguably
		simpler than the previous algorithms. We just guess the positions
		of certain jobs in the optimal solution, recurse on a set of guessed
		subintervals, and fill in the remaining jobs with greedy routines.
		We believe that also our analysis is more accessible, in particular since we do not
		use (LP-)hierarchies or abstractions of the problem like the ones above, but we guess properties
		of the optimal solution directly.
		
	\end{abstract}
	
	\section{Introduction}
	
	A classical problem in scheduling theory is the problem to schedule
	jobs on parallel machines in order to minimize the makespan, while
	obeying precedence constraints between the jobs. It goes back to the
	1966 when Graham proved in his seminal paper~\cite{graham1966bounds}
	that the simple List Scheduling algorithm yields a is a $(2-\frac{1}{m})$-approximation
	algorithm. Formally, the input consists of a set $J$ of $n$ jobs,
	a number of machines $m\in\N$, and each job $j\in J$ is characterized
	by a processing time $p_{j}\in\N$. We seek to schedule them non-preemptively
	on $m$ machines in order to minimize the time when the last job finishes,
	i.e., to minimize the makespan. Additionally, there is a precedence
	order $\prec$ which is a partial order between the jobs. Whenever
	$j\prec j'$ for two jobs $j,j'\in J$ then job $j'$ can only be
	started when $j$ has already finished. Given that the List Scheduling
	algorithm is essentially a simple greedy routine, one may imagine
	that one can achieve a better approximation ratio with more sophisticated
	algorithmic techniques. However, Svensson showed that even for unit
	size jobs (i.e., $p_{j}=1$ for each $j\in J$) there can be no $(2-\epsilon)$-approximation algorithm
	for any $\epsilon>0$~\cite{svensson2011hardness}, assuming a variant
	of the Unique Games Conjecture. Hence, under this conjecture List
	Scheduling is the essentially best possible algorithm. Slight improvements
	are known for unit size jobs: there is an algorithm by Coffman and
	Graham~\cite{coffman1972optimal} which computes an optimal solution
	when $m=2$, and a which is a $(2-\frac{2}{m})$-approximation algorithm
	for general $m$, as shown by Lam and Sethi~\cite{lam1977worst}.
	Also, there is an $(2-\frac{7}{3m+1})$-approximation algorithm known
	due to Gangal and Ranade~\cite{GANGAL20081139}. In fact, the setting
	of unit-size jobs is interesting since then the complexity of the
	problem stems purely from the precedence constraints and not from
	the processing times of the jobs (which might encode problems like
	\textsc{Partition}). We assume this case from now.
	% , i.e., we assume
	% that $p_{j}=1$ for each $j\in J$.
	
	In practical settings $m$ might be small, e.g., $m$ might be the
	number of processors in a system, or the number of cores of a CPU.
	Thus, it is a natural question whether better approximation ratios
	are possible when $m$ is a constant. Note that the mentioned lower
	bound of $2-\epsilon$~\cite{svensson2011hardness} does not hold
	if $m=O(1)$. For $m=2$ the mentioned algorithm by Coffman and Graham~\cite{coffman1972optimal}
	computes an optimal solution; however, even if $m=3$ it is not known
	whether the problem is NP-hard! In fact, it is one of the few remaining
	open problems in the book by Garey and Johnson~\cite{garey1979computers}.
	
	In a break-through result, Levey and Rothvoss presented a $(1+\epsilon)$-approximation
	with a running time of $n^{(\log n)^{O((m^{2}/\epsilon^{2})\log\log n)}}$~\cite{levey2019}.
	Subsequently, the running time was improved by Garg~\cite{garg:LIPIcs:2018:9063}
	to $n^{O_{m,\epsilon}(\log^{O(m^{2}/\epsilon^{2})}n)}$ which is quasi-polynomial.
	Both algorithms are based on the natural-LP relaxation of the problem,
	(essentially) lifted by a certain number $r$ of rounds of the Sherali-Adams
	hierarchy, and it can be solved in time $n^{O(r)}$. Given the optimal
	LP-solution $x^{*}$, they \emph{condition }on certain variables in
	the support of $x^{*}$, which effectively fixes time slots for the
	corresponding jobs. Each conditioning operation changes not only the
	variable that one conditions on, but possibly also other variables
	in the support. After a well-chosen set of conditioning operations,
	they recurse into smaller subintervals and give each of them a copy
	of the current LP-solution (which might be different from $x^{*}$
	due to the conditioning operations). Intuitively, in each recursive
	call for some subinterval $I$ they seek to schedule jobs that can
	only be scheduled during $I$ according to the previous conditionings
	and the precedence constraints;
	they call these jobs \emph{bottom jobs} and the other jobs \emph{top
	}and \emph{middle jobs}. The middle jobs can be discarded. For the
	top jobs, they first use a matching argument to show that most of
	the top jobs can be inserted if one can ignore the precedence constraints
	between top jobs. Knowing this, they insert most of the top jobs with
	a variation of Earliest-Deadline-First (EDF), such that all precedence constraints are satisfied;
	some of the top jobs are discarded in the process. The discarded jobs
	are later inserted in a greedy manner which is affordable since they
	are very few. 
	
	A different approach is used by Li~\cite{li2021towards} who improved
	the running time further to $n^{O_{m,\epsilon}(\log^{3}\log n)}$.
	Instead of working with an LP, he guesses directly certain properties
	of the optimal solution. While in the above argumentation each conditioning
	step costs a factor $n^{O(1)}$ in the running time, he argues that---roughly
	speaking---most of the time the information he guesses is binary
	and hence costs only a factor of 2 in the running time. More precisely,
	he guesses properties of a technical abstraction based on \emph{dyadic
		systems}, \emph{partial dyadic systems}, and \emph{virtually valid
		schedules}. On a high level, he shows that based on the optimal schedule
	one can define a corresponding dyadic system and a virtually valid
	schedule for it, at the cost of discarding a few jobs. His algorithm
	then searches for the dyadic system and a virtually valid schedule
	for it that discards as few jobs as possible. Then, he shows that
	based on them, he can construct an actually valid schedule, which
	discards only few additional jobs.

	\subsection{Our Contribution}
	
	In this paper we present a QPTAS for the makespan minimization problem with unit size jobs on a constant number of machines along with precedence constraints ($Pm|prec,p_{j}=1|C_{\max}$ in the three-field notation) which
	is arguably simpler than the $(1+\epsilon)$-approximation algorithms
	sketched above. We do not use an LP-formulation and in particular
	no LP-hierarchy or a similar approach based on conditioning on variables.
	Instead, we guess properties of the optimal solution directly, similarly
	as Li \cite{li2021towards}. However, we do not use the reduction
	to dyadic systems or a similar abstraction but work with the optimal
	solution directly. We believe that this makes the algorithm and the
	analysis easier to understand. Our running time is $n^{O_{m,\eps}((\log n)^{m/\eps})}$ so it is
	asymptotically better than the running time by Garg~\cite{garg:LIPIcs:2018:9063}
	up to hidden constants.
	
	Our algorithm is actually pretty simple. Let $T$ be the optimal makespan
	(which we guess). For a parameter $k=(\log n)^{O(1)}$ we guess the
	placement of $k$ jobs in $\OPT$ and a partition of $[0,T)$ into
	at most $k$ intervals. For each interval $I$, there are some jobs
	that need to be scheduled during $I$ according to our guesses and
	the precedence constraints. We recurse on each interval $I$ and its
	corresponding jobs. Then, we add all remaining jobs with a simple
	variant of EDF where the release dates and deadlines are defined such
	that the precedence constraints between these jobs and the other jobs
	(that we recursed on) are satisfied. For the correct guesses, we show
	that the resulting schedule discards at most $O(\epsilon T)$ jobs,
	and we add those jobs at the end with a simple greedy routine.
	
	In our analysis, we use a hierarchical decomposition of intervals,
	like the previous results~\cite{li2021towards,levey2019,garg:LIPIcs:2018:9063}.
	In contrast to those, we define the decomposition such that each interval
	is subdivided into $O(\log n/\eps)$ subintervals (instead of 2 subintervals)
	since this makes the analysis easier. Based on the optimal solution,
	for each level we identify certain jobs that we want to guess in that
	level later. Also, we assign a level to each job. Via a shifting step,
	we show that we can discard the jobs from intuitively every $(m/\epsilon)$-th
	level. Based on these levels, we argue that there are guesses for
	our algorithm that yield a small total number of discarded jobs. Intuitively,
	we guess the jobs from the first $m/\epsilon$ levels that we identified
	above and the intervals of the $(m/\epsilon+1)$-th level. We recurse
	on the latter intervals. When we insert the jobs of the first $m/\epsilon$
	levels via EDF (those jobs that we did not guess already), the jobs
	from our recursion and their precedence constraints dictate for each
	inserted job $j$ a time window $[r_{j},d_{j})$. It might happen
	that the position of $j$ in the optimal solution is not contained
	in $[r_{j},d_{j})$, but we show that it is always contained in the
	larger time window of
	$[r_{j}-\lambda ,d_{j}+ \lambda]$ for a small value $\lambda >0$.
	% $[r_{j}-\lambda T\frac{\epsilon}{\log n},d_{j}+T\frac{\epsilon}{\log n})$.
	We show that we need to discard at most $O(mT\frac{\epsilon}{\log n})$
	jobs to compensate for this error and due to the precedence constraints
	between the inserted jobs. We analyze EDF directly and do not need
	to go via a matching argumentation as done in~\cite{levey2019}.
	It turns out that our algorithm recurses for at most $O(\frac{\eps}{m}\log n)$ levels and hence this 
	yields $O(\frac{\eps\log n}{m}\cdot mT\frac{\epsilon}{\log n})=O(\epsilon^2 T)$ discarded
	jobs in total. 
	
	While our analysis borrows ideas from the mentioned previous results
	\cite{li2021towards,levey2019,garg:LIPIcs:2018:9063} we believe that
	our algorithm and our analysis are simpler and more accessible.

	\section{Algorithm}
	
	We present a simple QPTAS for the problem $Pm|prec,p_{j}=1|C_{\max}$.
	Let $\epsilon>0$ and assume w.l.o.g.~that $1/\epsilon\in\N$. We
	guess $T:=\OPT$ and assume w.l.o.g. that $T$ is a power of 2 (if
	this is not the case, then we can add some dummy jobs that need to
	be processed after all other jobs; note that a $(1+\epsilon)$-approximate
	solution for this larger instance is a $(1+2\epsilon)$-approximate
	solution for the original instance). Also, w.l.o.g. we assume that
	in $\OPT$ each job starts and ends at an integral time point, i.e.,
	during a time interval of the form $[t,t+1)$ for some $t\in\N$.
	We will refer to such a time interval as a \emph{time slot.} Furthermore, we assume that
	w.l.o.g. that the precendence constraints are transitive, i.e., if $j \prec j'$
	and $j' \prec j''$ then also $j \prec j''$.
	
	Our algorithm works on a `guess and recurse' framework. Define a parameter
	$k:=(m\log n/\eps)^{m/\epsilon+1}$. The reader can think of the parameter $k$ as $O_{m,\eps} ((\log n)^{m/\eps + 1})$, where $O_{m,\eps}()$ hides constants that are only dependent on $m, \eps$. Our algorithm has three steps. First,
	we guess up to $k$ jobs from $\OPT$ and their time slots in $\OPT$,
	i.e., we try all combinations of up to $k$ input jobs and all combinations
	for their time slots to schedule them. Then, we guess a partition
	of $[0,T)$ into at most $k$ intervals, i.e., we try all combinations
	of partitioning $[0,T)$ into at most $k$ intervals. By allowing
	empty intervals, we assume w.l.o.g. that we guess exactly $k$ intervals
	and denote them by $I_{1},...,I_{k}$. Let $J_{guess}$ denote the
	guessed jobs. There might be precedence constraints between jobs in
	$J_{guess}$ and jobs in $J\setminus J_{guess}$. In particular, these
	precedence constraints might dictate that some job $j\in J\setminus J_{guess}$
	needs to be scheduled within some interval $I_{j}$. In this case,
	we say that $j$ is a \emph{bottom job}. Let $J_{bottom}\subseteq J\setminus J_{guess}$
	denote the set of all bottom jobs. We call each job $j\in J\setminus(J_{guess}\cup J_{bottom})=:J_{top}$
	a \emph{top job. }
	
	Given our guess for the jobs $J_{guess}$ and their time slots and
	our guessed partitioning of $[0,T)$, we make $k$ recursive calls:
	one for each interval $I_{i}$ with $i=1,2,\cdots,k$. Let us denote
	by $J_{bottom}^{(i)}$ the subset of jobs in $J_{bottom}$ that need
	to be scheduled within $I_{i}$ according to our guesses above, and
	let $J_{guess}^{(i)}$ denote the jobs in $J_{guess}$ for which we
	guessed a time slot within $I_{i}$. We make a recursive call on the
	interval $I_{i}$ whose input are the jobs $J_{bottom}^{(i)}\cup J_{guess}^{(i)}$
	and the guessed time slots for the jobs in $J_{guess}^{(i)}$. In
	this recursive call, we want to compute a schedule in which 
	\begin{itemize}
		\item all jobs in $J_{guess}^{(i)}$ are scheduled in the time slots that
		we had guessed for them, 
		\item a (hopefully very large) subset of the jobs in $J_{bottom}^{(i)}$
		are scheduled; we denote by $J_{disc}^{(i)}\subseteq J_{bottom}^{(i)}$
		the jobs in $J_{bottom}^{(i)}$ that were not scheduled, we call them
		the \emph{discarded }jobs, and
		\item we obey the precedence constraints between the jobs in $J_{bottom}^{(i)}\cup J_{guess}^{(i)}$.
		We ignore all precedence constraints that involve the jobs in $J_{disc}^{(i)}$.
	\end{itemize}
	Suppose that we are given a solution from each recursive call. We
	define $J_{disc}:=\bigcup_{i=1}^{k}J_{disc}^{(i)}$. We ignore these
	discarded jobs for now. We want to schedule the jobs in $J_{top}$.
	Recall that these are jobs in $J\setminus J_{bottom}\cup J_{guess}$.
	To this end, for each job $j\in J_{top}$ we define an artificial
	release date $r_{j}$ and an artificial deadline $d_{j}$. Let $j\in J_{top}$.
	We define $r_{j}$ to be the
	earliest start time of an interval $I_i$ at which
	each job $j'\in J_{guess}\cup J_{bottom}$
	with $j'\prec j$ has completed; we define $r_{j}:=0$ if there is
	no such job $j'$. Similarly, we define $d_{j}$ to be the latest
	end time of an interval $I_i$ at which no job $j''\in J_{guess}\cup J_{bottom}$ with $j\prec j''$
	has already started; again, if there is no such job $j''$ we define $d_{j}:=T$.
	In order to find slots for the jobs in $J_{top}$ we use the following
	variation of the Earliest Deadline First (EDF) algorithm. We sweep
	the time axis from left to right. For each time $t=0,1,2,...$ we
	consider the not yet scheduled jobs $j\in J_{top}$ with $r_{j}\le t$
	whose predecessors in $J_{guess}\cup J_{bottom}\setminus J_{disc}$
	we have already scheduled before time $t$. We sort these jobs non-decreasingly
	by their deadlines (breaking ties arbitrarily) and add them in this
	order to the machines that are idle during $[t,t+1)$. We do this
	until no more machine is idle during $[t,t+1)$. It might happen that
	a job $j\in J_{top}$ misses its deadline at the current time $t$,
	i.e., it holds that $t=d_{j}$ but $j$ has not been scheduled by
	us at any slot $r_{j}\leq t'\leq t$ and during $[t',t'+1)$ all machines
	are busy. In this case we add $j$ to the set $J_{disc}$. Among all
	our guesses for the jobs and the partition into intervals, we output
	the solution in which at the very end the smallest number of jobs
	is in the set $J_{disc}$ (breaking ties arbitrarily).
	
	Each recursive call for an interval works similarly as the main call
	of the recursion described above. The (straightforward) differences
	are the following: the input of each recursive call consists of the interval
	$\bar{I}$, a set of jobs $\bar{J}_{guess}$ which were guessed in
	previous levels in the recursion, together with their guessed time
	slots, and a set of (not yet scheduled) jobs $\bar{J}$. We guess
	the time slots for up to $k$ guessed jobs $J_{guess}\subseteq\bar{J}$,
	and we require that these guesses do not violate the precedence constraints
	with the jobs in $\bar{J}_{guess}$. Also, we guess a partition of
	$I$ into $k$ intervals (rather than a partition of $[0,T)$). The
	input for each recursive call for a subinterval $\bar{I}_{i}$ of
	$\bar{I}$ consists of $\bar{I}_{i}$, $\bar{J}_{guess}\cup J_{guess}$,
	and the jobs in $\bar{J}$ that need to be scheduled within $\bar{I}_{i}$
	according to our guesses for $\bar{J}_{guess}\cup J_{guess}$. We
	return the computed schedule for $\bar{I}$ for a subset of the jobs
	$\bar{J}_{guess}\cup\bar{J}$ and the discarded jobs in $\bar{J}_{guess}\cup\bar{J}$.
	
	%introduce an empty
	%time slot during $[t,t+1)$. 
	%Formally, we shift the entire schedule
	%during $[t,\infty)$ by one unit to the right, and increase each deadline
	%within $[t,\infty)$ by one unit. Thus, by design of the algorithm
	%each job finishes by its (moved) deadline.
	
	If the algorithm is called on an interval of length 1 then we skip
	the step of partitioning the interval further into at most $k$ subintervals
	(and in particular we do not recurse anymore). We will show later
	that there are guesses that lead to an $(1+\epsilon)$-approximate
	solution such that the recursion depth is $\frac{\eps}{m}\log n$. In order to enforce
	that the recursion depth is $\frac{\eps}{m}\log n$ (limiting the running time
	of the algorithm), we define that a recursive call at recursion depth
	$\frac{\eps}{m}\log n+1$ simply outputs a solution in which all of the jobs in
	$\bar{J}$ are discarded and the algorithm does not recurse further.
	
	After running the recursive algorithm described above, we need to
	schedule the jobs in $J_{disc}$. Intuitively, for each such job $j\in J_{disc}$
	we create an empty time slot that we add into our schedule and inside
	which we schedule $j$. This will increase our makespan by $|J_{disc}|$.
	Formally, we consider the jobs $j\in J_{disc}$ in an arbitrary order.
	For each job $j\in J_{disc}$ we determine a time $t$ such that all
	its predecessors in our current schedule have finished by time $t$,
	but none of its successors in our current schedule have started yet.
	Such a time $t$ always exists since we show later that we obtain
	a feasible schedule for all jobs in $J \setminus J_{disc}$ and we assumed
	that the precedence constraints are transitive. We insert
	an empty time slot $[t,t+1]$ into our schedule, i.e., we move all
	jobs scheduled during $[t,\infty)$ by one unit to the right, and
	we schedule $j$ during $[t,t+1]$. This completes the description
	of our algorithm.

	\section{Analysis}
	
	In this section, we prove that the above algorithm is a QPTAS. 
	\begin{lemma}
		\label{lemma:runtime}
		The above algorithm runs in time $n^{(\log n)^{O_{m,\eps}(m/\epsilon)}}$ .
	\end{lemma} 
	\begin{proof} For the running time, we observe that
		there are at most $n^{O(k)}$ choices for the guessed $k$ jobs, at
		most $T^{O(k)}=(n)^{O(k)}=n^{O(k)}$ choices for their time slots,
		and similarly $n^{O(k)}$ choices for the at most $k$ intervals that
		we guess. Since we bound the recursion depth to be at most $\log n+1$,
		this yields a running time of $n^{O(k\log n)}=n^{(\log n)^{O_{m,\eps}(m/\epsilon)}}$.
	\end{proof} 
	Now we prove the approximation ratio of our algorithm.
	To this end, we define guesses for jobs and their time slots and intervals
	in each call of our recursion, such that for these guesses our algorithm
	outputs a schedule with makespan at most $1+6\epsilon$.
	
	\subsection{Laminar family of intervals}
	
	For this, we define a laminar family $\calL$ of intervals. Recall
	that we assumed that $T$, the guessed makespan, is a power of 2%
	\begin{comment}
	with a small increase in the optimal makespan (TODO: add a claim in
	Section 2)
	\end{comment}
	. We define that the entire interval $[0,T)$ forms the (only) interval
	of level 0. Let us define $\rho=\left\lceil \log(\log n/\eps)\right\rceil $.
	Consider an interval $I$ of some level $\ell=0,1,2,\cdots$ (we will
	argue the number of levels later). If $|I|\ge 2^\rho$ then $I$ is partitioned
	into $2^{\rho}=\Theta(\log n/\eps)$ equal-sized intervals of length $|I|/2^{\rho}$.
	These intervals constitute the level $\ell+1$ of the family $\calL$.
	For each interval $I\in\calL$, we denote by $\ell(I)$ the level
	of $I$.
	
	If $1<|I|<\rho$ then $I$ is partitioned into $|I|$ intervals of
	level $\ell+1$ of length 1 each. If $|I|=1$ then $I$ is not partitioned
	further.
	\begin{lemma}
		The total number of levels in the laminar family $\calL$ is at most
		$\log n/\log(\log n/\eps)+1$. %In particular,
		%for each interval $I$ of level $\ell=\log n/\log\log n + 2$ it holds
		%that $|I|=1$.
	\end{lemma}
	\begin{proof} By construction, each interval at a particular level
		$\ell=0,1,2,\cdots$ is of equal length. %The number of sub-intervals
		%created at level $\ell$ is at least $\log n$. 
		Hence the length of
		an interval at level $\ell$ is at most $T/(\log n/\eps)^{\ell}$. Further,
		once the length of intervals of a level becomes less than $2^{\rho}\le 2\log n/\eps$,
		there could only be one additional level where every interval is of
		length 1. Hence the total number of levels could be at most $\log T/\log(\log n/\eps)+1\leq\log n/\log(\log n/\eps)+1$
	\end{proof}
	
	\subsection{Guessed, top, and bottom jobs}
	\label{sec:guesses}
	Next, we assign the jobs to levels. More precisely, for each level
	$\ell$ we define a set of \emph{guessed jobs} $J_{guess}^{(\ell)}$
	and a set of \emph{top jobs }$J_{top}^{(\ell)}$. The intuition is
	that later we want to guess the jobs in $\bigcup_{\ell'=0}^{\ell} J_{guess}^{(\ell')}$
	and the jobs in $\bigcup_{\ell'=0}^{\ell} J_{top}^{(\ell')}$ will form
	top jobs. We say that a \emph{chain of jobs} is a set of jobs $J'=\{j_{1},j_{2},...,j_{c}\}$
	for some $s\in\N$ such that $j_{i}\prec j_{i+1}$ for each $i\in\{1,...,c-1\}$,
	and we say that $c$ is the \emph{length }of the chain. 
	
	\begin{comment}
	Let $I$ be the interval of level 0. We initialize $J_{guess}^{(0)}=J_{top}^{(0)}=\emptyset$.
	Our plan is that we add jobs to $J_{guess}^{(0)}$ step by step. These
	will later be jobs that we guess. We say that a job $j$ is \emph{flexible
	}if it could still be scheduled in more than one interval of level
	1, assuming that we schedule the jobs in $J_{guess}^{(0)}$ exactly
	as in $\OPT$. Suppose that there is a chain $J'\subseteq J\setminus J_{guess}^{(0)}$
	of length at least $T/2^{\left\lceil \log\log n\right\rceil }??$
	that contains only flexible jobs. Then for each interval $I'$ of
	level 1 we add to $J_{guess}^{(0)}$ the first and the last job from
	$J'$ that is scheduled during $I'$ in $\OPT$. If we guess these
	jobs, the effect is that each not guessed job in $J$ can be scheduled
	only during one interval of level 1. Hence, one way to think of this
	procedure is that we push these jobs down. We do this operation until
	there is no more chain $J'$ of length at least $T/2^{\left\lceil \log\log n\right\rceil }??$
	that contains only flexible jobs. We define that $J_{top}^{(0)}$
	contains all remaining flexible jobs. We say that we completely processed
	the jobs in level $0$. 
	\end{comment}
	
	We define these sets $J_{guess}^{(\ell)}$ and $J_{top}^{(\ell)}$
	level by level in the order $\ell=0,1,2,...$ . Consider a level $\ell$.
	Let $I$ be an interval of level $\ell$. We initialize $J_{guess}^{(\ell)}=J_{top}^{(\ell)}=\emptyset$.
	Our plan is that we add jobs to $J_{guess}^{(\ell)}$ step by step.
	Let $J_{I}$ denote the jobs that can only be scheduled during $I$,
	assuming that we schedule the jobs in $J_{guess}^{(0)}\cup\cdots \cup J_{guess}^{(\ell-1)}$
	exactly as in $\OPT$. We say that a job $j\in J_{I}$ is \emph{flexible
	}if it can still be scheduled in more than one subinterval of level
	$\ell+1$ of $I$, assuming that we schedule the jobs in $J_{guess}^{(0)}\cup...\cup J_{guess}^{(\ell)}$
	exactly as in $\OPT$. Suppose that there is a chain $J'\subseteq J_{I}\setminus J_{guess}^{(\ell)}$
	of length at least $\eps|I|/2^{\left\lceil \log\log n\right\rceil}$
	that contains only flexible jobs. Then for each interval $I'$ of
	level $\ell+1$ we add to $J_{guess}^{(\ell)}$ the first and the
	last job from $J'$ that is scheduled during $I'$ in $\OPT$. If
	we guess these jobs in our algorithm, the effect is that each job
	in $J'$ that we did \emph{not }add to $J_{guess}^{(\ell)}$ can be
	scheduled only during one interval of level $\ell+1$. Hence, one
	way to think of this procedure is that we push these jobs one level
	down. We do this operation until there is no more chain $J'$ of length
	at least $\eps|I|/m2^{\left\lceil \log\log n\right\rceil}$ that contains
	only flexible jobs. We define that $J_{top,I}^{(\ell)}$ contains
	all remaining flexible jobs in $J_{I}$. We do this procedure for
	each interval $I$ of level $\ell$, and define at the end $J_{top}^{(\ell)}:=\bigcup_{I}J_{top,I}^{(\ell)}$.
	\begin{proposition}
		\label{prop:unique} For every job $j\in J$, there exists a unique
		$\ell\in\{0,1,2,\cdots,\rho\}$ such that $j\in J_{top}^{\ell}$. 
	\end{proposition}
	
	\subsection{Few rejected jobs }
	
	With the preparation above, we will show that there are guesses of
	our algorithm for the guessed jobs and the intervals that lead to
	few discarded jobs overall, at most $O(\epsilon T)$ many. Since the
	algorithm selects the guesses that lead to the minimum total number
	of discarded jobs, we will show that it computes a solution with at
	most $O(\epsilon T)$ discarded jobs. %
	\begin{comment}
	shall now define a set of jobs and intervals that are potential guesses
	made by the algorithm. %In fact, each level of recursion in our algorithm will roughly correspond to $m/\eps$ levels in the laminar family. We shall call these collection of levels as \emph{superlevels}. 
	We will show later that for this specific set of candidate guesses
	(of jobs and intervals), there exists a feasible schedule of all but
	a set of at most $\epsilon T$ within time $T$. Since in each subcall
	of the recursion our algorithm selects the guesses that minimize the
	total number of jobs in $J_{disc}$ (SD : Check if the algorithm is
	really doing this), we will conclude that it outputs a solution in
	which at most $\epsilon T$ jobs are discarded. The final step of
	adding the jobs in $J_{disc}$ greedily, results in a schedule whose
	total length is at most $(1+\epsilon)T$.
	
	First, we show that 
	
	In order to keep the initial exposition simple and introduce the main
	ideas, we first consider the top level recursive call of the algorithm.
	In particular, recall that our algorithm guesses two things: up to
	$k=(\log n)^{m/\eps}$ jobs along with their correct time slots and
	a partitioning of the time horizon $[0,T]$ into up to $k$ subintervals.
	First we demonstrate a candidate set of such subintervals. 
	\end{comment}
	We need some preparation for this. First, we establish that we can
	afford to discard all jobs in sets $J_{top}^{(a+r\cdot m/\epsilon)}$
	for $r\in\N_{0}$, for some offset $a$.
	
	%First, with a shifting argument, we identify an offset $a\in\{0,1,...,\frac{m}{\epsilon}-1\}$
	%such that there are at most $\epsilon T$ jobs in the sets $\left\{ J_{top}^{(a+r\cdot m/\epsilon)}\right\} _{r\in\N}$. 
	%We will argue later that since they are so few jobs, we can afford
	%to make potentially very bad scheduling decisions for them, e.g.,
	%scheduling them greedily like the discarded jobs.
	\begin{lemma}\label{lemma:shift} There is an offset $a\in\{0,1,...,\frac{m}{\epsilon}-1\}$
		such that $\left|\bigcup_{r\in\N_{0}}J_{top}^{(a+r\cdot m/\epsilon + 1 )}\right|\le\epsilon T$.
	\end{lemma}
	
	\begin{proof} For every $a\in\{0,1,...,\frac{m}{\epsilon}-1\}$,
		we define $\calL_{a}=\{\ell:\ell=(a+r\cdot m/\epsilon + 1),r\in\N_{0}\}$
		and the set $\bigcup_{\ell\in\calL_{a}}J_{top}^{\ell}$. Now Proposition~\ref{prop:unique}
		implies that the resulting sets $\bigcup_{\ell\in\calL_{a}}J_{top}^{\ell}$
		are pairwise disjoint. Since $T$ is the optimal makespan, the total number
		of jobs cannot exceed $mT$. Hence, there exists some $a\in\{0,1,...,\frac{m}{\epsilon}-1\}$ such that
		$|\bigcup_{\ell\in\calL_{a}}J_{top}^{\ell}|\leq\eps T$. \end{proof}
	
	For the root problem of the recursion, we will show that the following
	guesses lead to few discarded jobs overall: the guessed subintervals
	are the subintervals of the laminar family at level $a+1$ where $a$
	which is the offset as identified by the above lemma. The guessed
	jobs are all jobs in $\bigcup_{\ell=0}^{a}J_{guess}^{(\ell)}$ and
	we guess their time slots in $\OPT$. We recurse on the guessed intervals
	of level $a+1$ of the laminar family. Suppose that in some level
	$r\in\N$ of the recursion we are given as input an interval $\bI$
	of some level $\ell_{r}=a+(r-1)\cdot m/\epsilon$ of the laminar family,
	together with a set of jobs of the form $\bar{J}=\bar{J}_{bottom}\dot{\cup}\bar{J}_{guess}$
	such that for each job $j\in\bar{J}_{guess}$ we are given a (guessed)
	time slot that equals the time slot in $\OPT$ during which $j$ is
	executed, but we are not given a time slot for any job in $\bar{J}_{bottom}$.
	We will show that the following guesses lead to few discarded jobs
	overall: we guess the intervals of level $\ell_{r+1}=$$a+r\cdot m/\epsilon+1$ %\syam{SD : As per the definition of $\ell_r$ above, this one should be  $\ell_{r+1} +1 $ ? }
	of the laminar family; the guessed jobs are all jobs in $\bar{J}_{bottom}\cap\bigcup_{\ell=a+(r-1)\cdot m/\epsilon + 1}^{a+r\cdot m/\epsilon}J_{guess}^{(\ell)}$
	that are scheduled in $\bI$ in $\OPT$ and we guess their time slots
	in $\OPT$. 
	
	With the next lemma, we prove inductively that there are few discarded
	jobs. We define $r_{\max}:=\left\lceil \epsilon(\log n/\log\log n+1)/m\right\rceil $
	which is an upper bound on the number of recursion levels that we
	need in this way.
	\begin{lemma}
		\label{lem:few-discard}Consider a recursive call of our algorithm
		in which the input is of the following form: 
		\begin{itemize}
			\item an interval $\bar{I}$ such that $\bar{I}=[0,T)$ (we define $r=0$
			in this case) or $\bar{I}$ is an interval of some level $\ell_{r}=a+(r-1)\cdot m/\epsilon$
			of the laminar family, for some $r\in\N,$
			\item a set of jobs $\bar{J}=\bar{J}_{bottom}\dot{\cup}\bar{J}_{guess}$
			such that 
			\begin{itemize}
				\item for each job $j\in\bar{J}_{guess}$ we are given a time slot that
				coincides with the time slot during which $j$ is scheduled in $\OPT$,
				\item each job $j\in\bar{J}_{bottom}$ is scheduled during $\bar{I}$ in
				$\OPT$.
			\end{itemize}
		\end{itemize}
		Then our algorithm returns a schedule for $\bar{J}$ in which at most
		\begin{equation}
		\sum_{\ell'\in\calL_{a}}\sum_{\bar{I}'\in\calL:\ell(\bar{I}')=\ell'\wedge\bar{I}'\subsetneq\bar{I}}\left|J_{top,\bar{I}'}^{(\ell')}\right|+\frac{5m\epsilon}{\log n}(r_{\max}-r)|\bar{I}|\label{eq:bound-discarded}
		\end{equation}
		jobs are discarded.
	\end{lemma}
	%\syam{SD : Do we need a lemma which says the number of jobs in guessed set is $k$ at each level of recursion}
	Our goal is now to prove Lemma~\ref{lem:few-discard}. Consider a
	recursive call of our algorithm of the form specified in Lemma~\ref{lem:few-discard}
	for some $r\in\N$. If $r=r_{\max}$ then our algorithm simply
	enumerates over all possible schedules for $\bar{J}_{bottom}$ and
	thus finds a schedule in which no job is discarded (since this is
	the case in $\OPT$). Suppose by induction that the claim is true
	for all $r\ge r^{*}+1$ for some $r^{*}$. We want to prove that it
	is true also for $r=r^{*}$ so we consider such a recursive call.
	Let $\tilde{J}_{guess}$ denote the guessed jobs and let $\tilde{I}_{1},...,\tilde{I}_{k}$
	denote the guessed partition of $\bar{I}$ into subintervals, according
	to our description right before Lemma~\ref{lem:few-discard}. Let
	$\tilde{J}_{top}\subseteq\bar{J}$ and $\tilde{J}_{bottom}\subseteq\bar{J}$
	denote the resulting set of top and bottom jobs, respectively (thus,
	the sets $\bar{J},\bar{J}_{guess},\bar{J}_{bottom}$ are part of the
	input, while the sets $\tilde{J}_{guess},\tilde{J}_{top},\tilde{J}_{bottom}$
	and intervals $\tilde{I}_{1},...,\tilde{I}_{k}$ stem from our guesses).
	Recall that for each job $j\in\tilde{J}_{top}$ we define a release
	time $r_{j}$ and a deadline $d_{j}$ in our algorithm. Let $\lambda$
	denote the length of each interval $\tilde{I}_{i}$ (note that they
	all have the same length), i.e., the length of the intervals of level
	$\ell_{r+1}=$$a+r\cdot m/\epsilon+1$ (where $r=0$ in the root problem
	of the recursion). Note that $\lambda \leq |\bI| / (\log n/\eps)^{m/\eps+1}$. 
	We show in the next lemma that in $\OPT$ each
	job $j\in\tilde{J}_{top}$ is essentially scheduled during $[r_{j},d_{j})$
	and thus $r_{j}$ and $d_{j}$ are almost consistent with $\OPT$.
	
	\begin{lemma} \label{lemma:like-OPT} For each job $j\in\tilde{J}_{top}$
		it holds that in $\OPT$ the job $j$ is scheduled during $[r_{j}-\lambda,d_{j}+\lambda)$.
	\end{lemma}
	
	\begin{proof} Let us recall that for each job $j\in\bar{J}_{top}$
		the release time $r_{j}$ is defined to be the earliest start time
		of an interval $\tilde{I}_{i},i=1,2,3,\cdots k$ such that every job $j'\in\bar{J}_{guess}\cup\bar{J}_{bottom}$
		with $j'\prec j$ is completed before $r_{j}$. We want to prove in
		$\OPT$ the job $j$ does not start before time $r_{j}-\lambda$.
		Let $\hI=[t_{1},t_{2}]$ denote the interval of level $\ell_{r+1}=$$a+r\cdot m/\epsilon+1$
		for which $t_{2}=r_{j}$ (and observe that $t_{2}=t_{1}+\lambda$).
		By definition of $r_{j}$, there exists a job $j'\in\bar{J}_{guess}\cup\tilde{J}_{guess}\cup\tilde{J}_{bottom}$
		with $j'\prec j$ that completes in $\hI$ in the schedule that we
		obtained from the recursive call in $\hI$. Since we assumed that
		our guessed time slots for the jobs in $\bar{J}_{guess}\cup\tilde{J}_{guess}$
		are identical to the corresponding time slots in $\OPT$, we conclude
		that also in $\OPT$ the job $j'$ is scheduled during $\hI$. Thus,
		in $\OPT$ the job $j$ cannot start before time $t_{1}=r_{j}-\lambda$.
		\begin{comment}
		such that $j'$ is a job corresponding to $\hI$ in $\bar{J}_{guess}\cup\bar{J}_{bottom}$
		and $j'\prec j$. By claim/lemma ?? (TODO: We definitely need a claim/lemma
		which states that the jobs in each subintervals of $\bI$ at every
		level of the laminar family is scheduled exactly as $\OPT$), $\OPT$
		must also schedule $j'$ precisely inside $\hI$. Hence, $j$ cannot
		be scheduled before time slot $r_{j}-|\hI|=r_{j}-\lambda$. 
		\end{comment}
		An analogous argument shows that $j$ cannot be scheduled in $\OPT$
		after $d_{j}+\lambda$.%(TODO: Write a lemma or a claim explicitly stating this).
		%AW: I adjusted this proof, please check!
	\end{proof}
	
	We want to show now that our variant of EDF discards only few jobs
	from $\tilde{J}_{top}$. To this end, we partition $\tilde{J}_{top}$
	into the sets $\tilde{J}_{top,1}:=\tilde{J}_{top}\cap \bigcup_{\ell=a+(r-1)\cdot m/\epsilon+1}^{a+r\cdot m/\epsilon}J_{top}^{(\ell)}$
	and $\tilde{J}_{top,2}:=\tilde{J}_{top}\cap J_{top}^{(a+r\cdot m/\epsilon + 1)}$.
	We can afford to discard all jobs in $\tilde{J}_{top,2}$, see~\eqref{eq:bound-discarded},
	but we need to bound the number of discarded jobs in $\tilde{J}_{top,1}$.
	For a job $j\in\tilde{J}_{top}$ it can happen that $r_{j}=d_{j}$.
	In this case we say that $j$ is \emph{degenerate}. Note that a degenerate
	job is always discarded. 
	\begin{lemma}
		\label{lem:deg}There are at most $2m\eps |\bI|/\log n$ degenerate jobs in $\tilde{J}_{top,1}$.
	\end{lemma}
	
	\begin{proof}
		Consider any job $j\in \tilde{J}_{top,1}$. By definition, there exist two adjacent intervals $I'=[t'_1, t'_2), I'' = [t'_2, t'_3)$ such that $\ell(I')=\ell(I'') = a+r\cdot m/\eps$ such that $j$ can be potentially scheduled during both the intervals as dictated by the guessed jobs $\bar{J}_{guess}$. Thus, if $j$ is degenerate, then $r_j = d_j = t'_2$. 
		Further, by Lemma~\ref{lemma:like-OPT}, all jobs $j\in \tilde{J}_{top,1}, r_j=d_j=t'_2$ must be scheduled in $\OPT$ in the interval $[t'_2-\lambda, t'_2+\lambda]$. Hence the total number of such jobs is upper bounded by
		
		\begin{align*}
		& m\cdot 2\lambda|\{I' : I'\in \calL \wedge \ell(I') =a +r\cdot m/\eps  \wedge I'\subset \bI \}| \\
		& \le 2m\cdot (\log n/\eps)^{m/\eps}\cdot \frac{|\bar{I}|}{(\log n/\eps)^{m/\eps +1}} \\
		& = 2m\eps|\bI|/\log n 
		\end{align*}

		% By Lemma~\ref{lemma:like-OPT}, $\OPT$ schedules $j$ in the interval $[t'-\lambda, t'+\lambda]$. Since the total number of subintervals of level $a+r\cdot m/\eps$.
	\end{proof}
	
	\begin{comment}
	Unfortunately, we cannot bound the number of degenerate jobs in $\tilde{J}_{top,2}$
	which is why we need Lemma~\ref{lemma:shift} for the jobs in $\tilde{J}_{top,2}$. 
	\end{comment}
	Our goal now is to bound the number of discarded non-degenerate jobs
	in $\tilde{J}_{top,1}$. We partition $\bar{I}$ into \emph{meta-intervals}
	$\hat{I}_{1},\hat{I}_{2},...\hat{I}_{k'}$ with $k'\le k$ with the
	properties that each interval $\hat{I}_{i}\in\{\hat{I}_{1},\hat{I}_{2},...\}$
	is of the form $\hat{I}_{i}=[t_{1},t_{2})$ for some $t_{1},t_{2}\in\N$
	such that 
	\begin{itemize}
		\item each value $t_{1},t_{2}$ is the start or the end point of some interval
		in $\tilde{I}_{1},...,\tilde{I}_{k}$, 
		\item at time $t_{2}$ there is no (non-degenerate) job $j\in\tilde{J}_{top}$
		pending that was released before $t_{2}$,
		\item for each $t\in[t_{1},t_{2})$ such that $t$ is the start or end point
		of some interval in $\tilde{I}_{1},...,\tilde{I}_{k}$, some non-degenerate
		job $j\in\bar{J}_{top}$ is pending at time $t$. 
	\end{itemize}
	Now the intuition is that during each meta-interval $\hat{I}_{i}=[t_{1},t_{2})$
	EDF tries to schedule only jobs that $\OPT$ schedules during $[t_{1}-\lambda,t_{2}+\lambda)$
	(due to Lemma~\ref{lemma:like-OPT}), so essentially we have enough
	space on our machines to schedule all these jobs. We might waste space
	due to the precedence constraints. However, this space is bounded
	via the following lemma. 
	
	\begin{lemma} \label{lemma:waste}Let $\hat{I}_{i}$ be a meta-interval
		such that at each time $t\in\hat{I}_{i}$ some job $j\in\tilde{J}_{top}$
		is pending. Then during $\hat{I}_{i}$ there are at most $\left|\hat{I}_{i}\right|\frac{\epsilon}{m\log n}$
		time slots $[t,t+1)$ with $t\in\N$ such that some machine is idle
		during $[t,t+1)$. \end{lemma}
	\begin{proof}
		For the meta-interval $\hI_i = [ t_1, t_2 )$, let $J_{\hI_i}$ denote the subset of 
		jobs in $j\in \bar{J}_{top,1}$ such that $r_j \leq t_1$. We now create a partition of $J_{\hI_i}$ according to the precedence constraints. Let $J_0 \subseteq J_{\hI_i}$ denote the jobs whose preceding jobs have been either scheduled or discarded before $t_1$. For every $p = 1,2,\cdots \eta$ (where $\eta$ is some positive integer), let $J_p$ denote the set of jobs $j\in J_{\hI_i}$ for which there exists a job $j'\in J_{p-1}$ such that $j' \prec j$ and there is no job $j''\in J_{\hI_i}\setminus J_{p-1}$ such that $j'' \prec j$.
		
		Let $\tau_1, \tau_2, \tau_3, \cdots , \tau_{\eta'}$
		be defined such that the time slots $[\tau_1, \tau_{1}+1), [\tau_2, \tau_{2}+1),...,[\tau_{\eta'}, \tau_{\eta'}+1) $
		are exactly the time slots during 
		% denote the time slots $[\tau_p, \tau_{p+1} )$
		$\hI_i$ such that some machine is idle for the entire respective time slot. We claim that any job $j'\in \tilde{J}_{top}$ that is pending at the end of a time slot
		$[\tau_q, \tau_{q}+1)$ with
		$q=\{1,2,3,\cdots \}$ must belong to $J_p$ for some $p > q$. We prove this by induction on $q$. For the base case, since no jobs are released inside $\hI_i$, no job in $J_0$ is pending at time $\tau_0 +1$ since otherwise this would contradict the definition of our variant of EDF and the fact that a machine is idle during $[\tau_1, \tau_{1}+1)$. Now assume the hypothesis to be true for the time slots
		$[\tau_1, \tau_{1}+1), [\tau_2, \tau_{2}+1),...,[\tau_q, \tau_{q}+1)$
		for some $q$ and consider the time slot $[\tau_{q+1}, \tau_{q+1}+1)$.
		If a job $j'\in \tilde{J}_{top}$ is pending at time $\tau_q +1$ then $j' \in J_{p'}$ for some $p' > q$ by the induction hypothesis. This means that all jobs in $\bigcup_{p=0}^{q} J_p$ have completed before time $\tau_q +1$. Hence, the jobs in $J_{q+1}$ can be scheduled at any time after time $\tau_q +1$.
		Suppose that there is a pending job $j$ at time $\tau_{q+1} +1$. Since a machine is idle during $[\tau_{q+1}, \tau_{q+1}+1)$
		we have that $j\in J_{p'}$ for some $p' > q+1$ since otherwise our variant of EDF would have scheduled $j$ during $[\tau_{q+1}, \tau_{q+1}+1)$.
		%
		% which could not have been released inside $\hI_i$ by fact that $\forall j\in \tilde{J}_{top}$, $r_j$ is the start time of some interval and level $\ell_{r+1}$. Further, a machine is idle for the entire duration of the time slot $\tau_{q+1}$. The pending job at the end of $\tau_{q+1}$ then must belong to $J_{p'}, p' > q+1$.
		Now consider any interval $\tilde{I}$ contained within the meta-interval $\hat{I}$. By our construction of the guessed, top and bottom jobs in Section~\ref{sec:guesses} the length of the longest chain of the jobs in $\tilde{J}_{top}$ is
		at most
		$\eps|\tilde{I}|/2^{\left\lceil \log\log n\right\rceil }\le\left|\hat{I}_{i}\right|\frac{\epsilon}{\log n}$, since $|\hat{I}_i|$ is clearly an upper bound on $\tilde{I}$.
		Therefore, $\eta' \le \eta \le \left|\hat{I}_{i}\right|\frac{\epsilon}{\log n}$ which completes our proof.
		%
		% The above claim implies that if $\eta' \ge \left|\tilde{I}_{i}\right|\frac{\epsilon}{\log n}$
		%
		% there is a time slot $\tau_{\left|\tilde{I}_{i}\right|\frac{\epsilon}{\log n}}$, then the pending job at this time slot must belong to $J_{p'}, p' > \left|\tilde{I}_{i}\right|\frac{\epsilon}{\log n}$. This further implies that there exists a chain of jobs in $\tilde{J}_{top}$ of length larger than $\left|\tilde{I}_{i}\right|\frac{\epsilon}{\log n}$. This is a contradiction since by our construction of the guessed, top and bottom jobs in Section~\ref{sec:guesses}, we would have broken this chain by guessing the first and last job of this chain inside the sub-interval $\hI_i$.
		% \syam{SD: I feel we need to refer to the place where we are doing this. But there is no lemma which says this?}  
	\end{proof}
	
	Also, we show that there are only few meta-intervals which---together
	with the argumentation above---yields the following lemma.
	
	\begin{lemma} \label{lem:non-deg} The total number of discarded
		non-degenerate jobs in $\tilde{J}_{top,1}$ is at most $m|\bI|\left(\frac{2\eps}{\log^2 n} + \frac{\eps}{\log n}\right)$.
	\end{lemma}
	\begin{proof}
		Consider a meta-interval $\hat{I}_{i}$. Let $j$ be the last non-degenerate
		job in $\bar{J}_{top,1}$ that is discarded during $\hat{I}_{i}$.
		Let $t_{1}$ denote the beginning of the interval $\hat{I}_{i}$.
		We know that $d_{j}$ is the end point of some interval $\tilde{I}_{i'}$.
		Since $j$ is non-degenerate we know that $[r_{j},d_{j})\ne\emptyset$
		and by definition of $\hat{I}_{i}$ this implies that at each time
		$t\in[t_{1},d_{j})$ there is some job from $\bar{J}_{top}$ pending.
		Using Lemma~\ref{lemma:waste}, the total number of (wasted) idle slots across
		all machines during $[t_{1},d_{j})$ is at most $m\eps|[t_{1},d_{j})|/m\log n$.
		%AW: why is ``which is an upper bound to the number of discarded jobs
		%whose artificial release dates and deadlines fall in the subinterval
		%$[t_{1},d_{j})$.'' correct?
		
		We further invoke Lemma~\ref{lemma:like-OPT} to argue that all jobs
		in $\tilde{J}_{top}$ that we schedule or discard during $[t_{1},d_{j})$
		are scheduled in $\OPT$ during $[t_{1}-\lambda,d_{j}+\lambda)$.
		The maximum number of jobs that $\OPT$ could have scheduled during
		$[t_{1}-\lambda,d_{j}+\lambda)$ is hence $m(d_{j}-t_{1}+2\lambda)$.
		Since our wasted space during $[t_{1},d_{j})$ is at most $m\eps|[t_{1},d_{j})|/m\log n$,
		we conclude that we discard at most 
		\[
		m(d_{j}-t_{1}+2\lambda)-m(d_{j}-t_{1})+\frac{\epsilon}{\log n}(d_{j}-t_{1})=2\lambda m+\frac{\epsilon}{\log n}(d_{j}-t_{1})
		\]
		jobs during $[t_{1},d_{j})$. By definition, for each job $j\in\bar{J}_{top,1}$
		we have that $j\in J_{top}^{(\ell)}$ for some level $\ell$ with
		$\ell\le a+(r+1)\cdot m/\epsilon-1$. Thus, for such a job $j$ we
		have that $r_{j}$ and $d_{j}$ lie in different intervals of level
		$a+(r+1)\cdot m/\epsilon$ of $\calL$. Thus, if during a meta-interval
		$\hat{I}_{i}$ a job $j\in\bar{J}_{top,1}$ is discarded, then $\hat{I}_{i}$
		has non-empty intersection with at least two different intervals of
		level $a+(r+1)\cdot m/\epsilon$ of $\calL$. Hence, there can be
		at most $(\log n/\eps)^{m/\epsilon-1}$ such meta-intervals. We conclude
		that in total we discard at most 
		\[
		(\log n/\eps)^{m/\epsilon-1}\cdot 2\lambda m+\frac{\epsilon}{\log n}|\bar{I}|\le m\frac{2|\bI|(\log n/\eps)^{m/\epsilon-1}}{(\log n/\eps)^{m/\eps +1}} +\frac{\epsilon}{\log n}|\bar{I}| \le m|\bI|\left(\frac{2\eps}{\log^2 n} + \frac{\eps}{\log n}\right)
		\]
		jobs in $\bar{J}_{top,1}$.%
		\begin{comment}
		\begin{proof}
		
		the level of $j$ is at most $a+(r+1)\cdot m/\epsilon-1$. In particular,
		$r_{j}$ and $d_{j}$ lie in a different subintervals of level $a+(r+1)\cdot m/\epsilon$
		of $\calL$. Thus, if during a meta-interval $\hat{I}_{i}$ a job 
		
		$j\in$
		
		these two subintervals can be $2m|\hat{I}_{i}|/\log n$. Hence, the
		total number of discarded job whose release dates and deadlines are
		in the meta-interval $\tilde{I}_{i}$ can be at most $(\eps+1/\log n)m|[t_{1},d_{j})|\leq(\eps+1/\log n)m|\hat{I}_{i}|$.
		
		Extending the above argument to all meta-intervals $\tilde{I}_{i},i=1,2,\cdots k'$
		and observing that they are disjoint, the total number of discarded
		jobs from $\bar{J}_{top,1}$ is at most
		
		\[
		\sum_{i=1}^{k'}m(\eps+1/\log n)|\tilde{I}_{i}|=m(\eps+1/\log n)|\bI|
		\]
		\end{proof}
		\end{comment}
	\end{proof}
	Now we are ready to bound the total number of discarded jobs using Lemmas~\ref{lem:deg},
	\ref{lem:non-deg}, and the induction hypothesis.
	\begin{proof}[Proof of Lemma~\ref{lem:few-discard}]
		We shall prove the lemma for the input interval $\bI$ at level $\ell_{r^{*}} = a + (r^{*}-1)\cdot m/\eps$.
		% The total number of discarded jobs is bounded by $???\,TODO:...calculation...$
		By induction hypothesis, suppose~\eqref{eq:bound-discarded} is true for all $r \ge r^{*} + 1$. Hence, applying our recursive algorithm with the input intervals $\tilde{I}_i, i=1,2,\cdots k$ returns a schedule where number of discarded jobs is at most
		
		\begin{equation}
		\label{eq:ind-hyp}
		\sum_{\ell'\in\calL_{a}}\sum_{\bar{I}'\in\calL:\ell(\bar{I}')=\ell'\wedge\bar{I}'\subsetneq\tilde{I}_i}\left|J_{top,\bar{I}'}^{(\ell')}\right|+\frac{5m\epsilon}{\log n}(r_{\max}-(r^{*} + 1)|{\tilde{I}_i}|
		\end{equation}
		
		Summing~\eqref{eq:ind-hyp} over all $i=1,2,\cdots k$ and observing that $\bI = \dot\bigcup_{i=1}^k \tilde{I}_i$ yields that the total number of discarded jobs is at most
		
		\begin{equation}
		\label{eq:ind-hyp-sum}
		\sum_{\ell'\in\calL_{a}}\sum_{\bar{I}'\in\calL:\ell(\bar{I}')=\ell'\wedge\bar{I}'\subsetneq\bar{I}}\left|J_{top,\bar{I}'}^{(\ell')}\right|+\frac{5m\epsilon}{\log n}(r_{\max}-(r^{*} + 1))|{\bar{I}}|
		\end{equation}

		We would now like to prove the statement for $r=r^{*}$. Now at the recursive call at level $r=r^{*}$ with the input subinterval $|\bI|$ consider the guesses as described in the preceding discussion. Using Lemmas~\ref{lem:non-deg} and~\ref{lem:deg}, the total number of jobs rejected from $\tilde{J}_{top,1}$ under these guesses is at most 
		
		\begin{equation}
		\label{eq:top1}
		2\eps m|\bI|/\log n + m|\bI|\left(\frac{2\eps}{\log^2 n} + \frac{\eps}{\log n}\right) \leq 5m\eps|\bI|/\log n 
		\end{equation}
		
		Further, we could have potentially rejected all the jobs in $\tilde{J}_{top,2}$. The total number of such jobs is 
		
		\begin{equation}
		\label{eq:top2}
		|\tilde{J}_{top,2}| = \sum_{\bI'\in \calL:\ell(\bI')=\ell', \ell'=a+r^{*}\cdot m/\eps} \left| J^{(\ell')}_{top, \bI} \right| 
		\end{equation}
		
		Adding the above two quantities~\eqref{eq:top1} and~\eqref{eq:top2} to the quantity~\eqref{eq:ind-hyp-sum} and observing that our recursive algorithm selects the guesses at a particular level that minimizes the number of discarded jobs, the lemma holds for a recursive call at level $r^{*}$.
		
	\end{proof}

	\begin{lemma}
		\label{lemma:makespan}	
		Our algorithm computes a solution with a makespan
		of at most $(1+6\epsilon)T$. \end{lemma}
	
	\begin{proof}
		
		The input at the top level recursive call in our algorithm is an interval $\bI = [0, T)$ and the entire set of jobs $J$. Plugging in $r=0$ in Lemma~\ref{lem:few-discard} and using Lemma~\ref{lemma:shift},  the first term in~\eqref{eq:bound-discarded} is bounded by $\sum_{\ell'\in \calL_a} |J_{top}^{\ell} | \le \eps T$
		.
		Since  $r_{\max}:=\left\lceil \epsilon(\log n/\log\log n+1)/m\right\rceil $, the second term for $r=0$ in~\eqref{eq:bound-discarded} bounded by $5\eps^2 T/(\log\log n +1)$. 
		
		Hence, the number of discarded jobs in our algorithm is upper bounded by $6\eps T$. All the other jobs are scheduled within the interval $[0, T)$. We potentially need to introduce  one additional time-slot of the form $[t, t+1)$ for each discarded job. We observe that for each discarded job $j$ we can find
		a position to insert a time-slot for $j$ since we assumed the precendence constraints to be transitive and we obtained a feasible schedule for all
		non-discarded jobs. Hence, the total length of the schedule is at most $(1+6\eps)T$.
	\end{proof}  
	
	Combining Lemma~\ref{lemma:makespan} and Lemma~\ref{lemma:runtime} gives us the following theorem. 
	
	\begin{theorem}
		There exists an algorithm for the precedence constrained scheduling on identical parallel machines that is a $(1+\eps)$-approximation and runs in time $n^{O_{m,\eps}((\log n)^{m/\eps})}$. 
	\end{theorem}
	
	\bibliographystyle{plain}
	\bibliography{bibliography}
	
\end{document}